\documentclass[journal]{IEEEtran}

\makeatletter
\def\ps@headings{%
\def\@oddhead{\mbox{}\scriptsize\rightmark \hfil \thepage}%
\def\@evenhead{\scriptsize\thepage \hfil \leftmark\mbox{}}%
\def\@oddfoot{}%
\def\@evenfoot{}}
\makeatother
\usepackage{cite}
\usepackage{epstopdf}
\usepackage{graphicx}
\usepackage{algorithm,algpseudocode}
\usepackage[cmex10]{amsmath}
\usepackage[margin=10pt,font=small,labelfont=bf]{caption}
\usepackage{subcaption}
\usepackage{enumerate}
\usepackage{bbold}
\usepackage{bm}
\usepackage{rotating}
\usepackage{multirow}
\usepackage{etoolbox}
\usepackage[usenames,dvipsnames]{xcolor}
\usepackage{array}

\newtoggle{conference}
\togglefalse{conference} 

\graphicspath{{figures/}}

\def\beq{\begin{equation}}
\def\eeq{\end{equation}}
\def\beqa{\begin{eqnarray}}
\def\eeqa{\end{eqnarray}}
\def\beqan{\begin{eqnarray*}}
\def\eeqan{\end{eqnarray*}}

\def\argmin{\mathop{\mathrm{arg\,min}}}
\def\argmax{\mathop{\mathrm{arg\,max}}}

\newtheorem{lemma}{Lemma}

\def\rbar{\overline{r}}
\def\rhat{\widehat{r}}

\def\what{\widehat{w}}

\def\zhat{\widehat{z}}

\def\arr{\rightarrow}

\def\tm1{t\! - \! 1}
\def\tp1{t\! + \! 1}

\def\bbf{\mathbf{b}}

\def\rbf{\mathbf{r}}
\def\rbfhat{\widehat{\mathbf{r}}}
\def\rbfbar{\overline{\mathbf{r}}}

\def\wbf{\mathbf{w}}
\def\wbfhat{\widehat{\mathbf{w}}}

\def\zbf{\mathbf{z}}

\def\zbfhat{\widehat{\mathbf{z}}}
\def\Abf{\mathbf{A}}

\def\Gbf{\mathbf{G}}

\def\Sbf{\mathbf{S}}

\def\mubf{{\boldsymbol \mu}}
\def\lambdabf{{\boldsymbol \lambda}}

\def\thetabf{{\bm{\theta}}}
\def\thetabfhat{{\widehat{\bm{\theta}}}}

\begin{document}

\title{Opportunistic Third-Party Backhaul for Cellular
Wireless Networks}
\iftoggle{conference}
{ 

    \author{
        \IEEEauthorblockN{Russell Ford, Changkyu Kim,
        Sundeep Rangan\\}
        \IEEEauthorblockA{Polytechnic Institute of New York University,
        Brooklyn, New York\\
        Email: \{rford02,ckim02,yqi\}@students.poly.edu, srangan@poly.edu}
    }
}{ 
    \author{
        Russell Ford,~\IEEEmembership{Student~Member,~IEEE},
        Changkyu Kim,~\IEEEmembership{Student~Member,~IEEE},
        Sundeep Rangan,~\IEEEmembership{Senior Member,~IEEE}
        \thanks{This material is based upon work supported by the National Science
        Foundation under Grant No. 1116589.}
        \thanks{R. Ford (email:rford02@students.poly.edu) and
                S. Rangan (email: srangan@poly.edu) are with the
                Polytechnic Institute of New York University, Brooklyn, NY.}
    }
}

\maketitle

\begin{abstract}
With high capacity air interfaces and large numbers of small cells,
\emph{backhaul} -- the wired connectivity to base stations
-- is increasingly becoming the cost driver in cellular wireless networks.
One reason for the high cost of backhaul is that
capacity is often purchased on leased lines with guaranteed rates
provisioned to peak loads.
In this paper, we present an alternate \emph{opportunistic backhaul} model where
third parties provide base stations and backhaul connections and lease
out excess capacity in their networks to the cellular provider
when available, presumably at significantly
lower costs than guaranteed connections.
We describe a
scalable architecture for such deployments using open access \emph{femtocells},
which are small plug-and-play base stations that operate
in the carrier's spectrum but can connect directly
into the third party provider's wired network.
Within the proposed architecture, we present a general
user association optimization algorithm that enables the cellular provider
to dynamically determine which mobiles
should be assigned to the third-party femtocells based on the
traffic demands, interference and channel conditions
and third-party access pricing. Although the optimization is non-convex,
the algorithm uses a computationally efficient method for finding
approximate solutions via dual decomposition.
Simulations of the deployment model
based on actual base station locations are presented that
show that large capacity gains are achievable
if adoption of third-party, open access femtocells can reach even a
small fraction of the current market penetration of WiFi access points.
\end{abstract}

\iftoggle{conference}{}{
}

\section{Introduction}

Cellular wireless networks have been traditionally designed on the premise that
the wireless interface is the bottleneck for system throughput and capacity.
However, a surprising recent trend is that \emph{backhaul}, meaning the wired
connectivity to the base stations,
is increasingly becoming the dominant cost
driver in many networks~\cite{backhaul_blog_1,backhaul_blog_2}.
Even for comparatively lower data rate pre-4G systems, backhaul
already accounted for a significant percentage of the operating
costs (30 to 50\% by some estimates \cite{financial_picos,senza_backhaul}).
Higher data rate 4G systems combined with the increasing adoption of a large numbers
of small cell deployments~\cite{financial_picos} will require even greater
costs in the backhaul,
particularly in markets where the operator does not have universal fiber access.

This work presents a novel deployment model for cellular providers that would enable the rising costs of backhaul networks to be mitigated by offloading traffic to third-party backhaul connections.
The basic premise is that backhaul services are currently
purchased with guaranteed service level agreements (SLAs) along dedicated lines \cite{ChiaGB:09}, which come at a significant cost for operators.
These SLAs must generally be provisioned for the \emph{peak} data rates.
However, due to variations in loading and channel conditions, much of the purchased
capacity goes to waste.

We propose that, rather than provisioning these links for the peak demand, cellular networks should be able to dynamically leverage excess capacity on existing backhaul links provided by third-party entities. The role of this third party may be played by other services providers (i.e.\ wireline ISPs) or even broadband customers, the very end-users themselves.

The third parties can provide connectivity to the operator's subscribers through \emph{femtocells}~\cite{ChaAndG:08,LopezVRZ:09,AndrewsCDRC:12}, which
are small, low-cost, cellular base
stations that operate in the provider's spectrum
but are connected into the third party's backhaul.
The network can then offload mobile subscribers
to the third party femtocells and the cellular provider would reimburse the
third party for use of the backhaul resources (and possibly cover the one-time
cost of the femtocell as well).  The key is to offload traffic \emph{opportunistically}
when third parties have excess backhaul capacity.
Since this capacity would only be purchased when used and since mobile traffic
would generally represent only a small increment in average demand
at most enterprises and residences, the opportunistic capacity can presumably
be purchased at much a lower cost than guaranteed lines to base stations.

In addition, significant progress has been made in making femtocells
completely self-organizing with ``plug-and-play" installation
\cite{ZhangRoche:10,SONbook:11},
implying that third-party femtocells would have minimal operating costs.
In this way, opportunistic backhaul with third-party open-access
femtocells can provide a scalable model for high-density,
high-capacity deployments at low cost.

We describe a potential architecture for third-party opportunistic
backhaul model within the LTE/SAE framework
\cite{OlssonSRFM:09} (see Section~\ref{sec:architecture}).
Within this architecture, we consider one of the key technical problems, namely
the optimization of \emph{user association:}
The cellular provider's network
must dynamically assign mobile subscribers between the
operator-controlled base stations and third party femtocells
based on channel and interference conditions, backhaul capacity,
traffic demands and third party access pricing.
We present a general optimization formulation
for this problem using a recently-developed
methodology in \cite{KimFQR:13-arxiv}, which itself was based on \cite{YeRon:12}.
The user association optimization
problem is generally non-convex, but following \cite{KimFQR:13-arxiv},
we show that the optimization admits a dual decomposition that
enables application of efficient approximate augmented Lagrangian
methods.
The methodology is extremely general and
enables joint optimization for load balancing and
interference coordination and can incorporate a large class of interference
models, network topologies and pricing schemes.

To evaluate the potential
capacity increase for this opportunistic backhaul model,
we present a simple simulation
where we assume open access femtocells can be
co-located at a small fraction (between 2 and 25\%)
of the locations where current residential and enterprise WiFi access points
are deployed. Basing our simulations on reported location data of WiFi access points and cellular
base stations and industry standard cellular evaluation models \cite{3GPP36.814},
we show that large capacity gains are possible with offload into
third-party networks and we discuss the potential costs of investment in the third-party model compared to the additional operator-deployed or leased line infrastructure required to support an equivalent gain in system throughput.

\subsection{Related work}

Although femtocells have been traditionally used for
improving coverage in private residences and enterprises
\cite{ChaAndG:08,LopezVRZ:09,AndrewsCDRC:12},
open access femtocells deployed explicitly for wide-area coverage
have also been considered -- see, for example, Qualcomm's
\emph{neighborhood small cells} whitepaper \cite{Qualcomm-NSC}.
That work showed significant capacity gains would be possible
with open-access LTE femtocells placed at a small fraction (10\%)
of residences.  The analysis, however, does not explicitly consider
the issues of backhaul usage and third-party charging, which are
the main focus of this paper.

A comparable business model is also found in public WiFi networks
such as FON, which boasts itself as the ``world's largest WiFi network'' \cite{fon}. FON members, called ``Foneros", agree to allow other Foneros to securely connect to their custom home WiFi AP and, in return, gain access the millions of other hotspots hosted by members of the community. Under one of the available membership plans, users are incentivized to provide reliable WiFi service by being compensated per-byte of data traffic consumed by other Foneros. The FON business model has been tremendously successful in recent years, bringing in 28 million Euro in revenue during 2010~\cite{techcrunch_fon}.  In this work, we consider offload to cellular,
rather than WiFi, which has the advantage of better support for mobility and
interference management.  Moreover, in the model presented, the network
makes all cell selection decisions and is responsible for paying third parties
for offloaded traffic; third-party vs.\ provider ownership is transparent to the mobile
subscribers who see all base stations as belonging to a single network.

A mathematical evaluation of pricing schemes for mobile subscribers is presented in \cite{yun_femto_econ}, which concludes that operators can maximize revenue by offering femto services to all customers at a flat rate, that is, not as an extra value-added service but part of the basic package. We adopt this method of subscriber charging in our model, which we believe has the added benefit of encouraging femtocell adoption. The utility for data consumers is often modeled by a logarithmic function of rate (a special case of iso-elastic utility), which has the rational basis of expressing the decreasing marginal payoff of rate experienced by users~
\cite{KelleyMT:98,dasilva_pricing,elasticity}.

Understanding the shape of the supply-side utility curve for third-party backhaul providers is not as straightforward, however. Intuitively, a monopolist operator wants to offer a price for backhaul capacity that maximizes their net utility, which is a function of the average or total throughput seen by users as well as the cost of connectivity over third-party links. Although we consider the operator to be a price taker under this model, if the third-party provider is a individual end-user with leased broadband capacity, we assume their supply-side price elasticity is likely to be very inelastic since they will accept any price offered for their excess bandwidth (seeing as how they would otherwise get nothing) \cite{elasticity}. This dynamic becomes significantly more complex, requiring a game-theoretic approach to analysis, when we consider a competitive market where one or more broadband ISPs offer resources to one or more mobile service providers and the price of said resources may be a time-varying function of demand.\footnote{Such multi-sided markets are the subject of \cite{multi_provider_bidding}, which introduces a bidding system through which providers compete for network resources.}

In this paper, we forgo an analysis of competitive markets in favor of investigating the relationship between a single cellular provider (with existing macrocellular infrastructure) and many 3rd-party backhaul providers. We assume this baseline model in order to demonstrate a general framework for determining net utility gains and couch an upper bound on the incentive that could be offered while still increasing operator revenue.

\section{System Architecture} \label{sec:architecture}
\begin{center}
\label{sec:cell_selection}
\begin{figure}
\centering
\includegraphics[width=6.0in,trim=1.0in 2.8in 0in 3.0in]{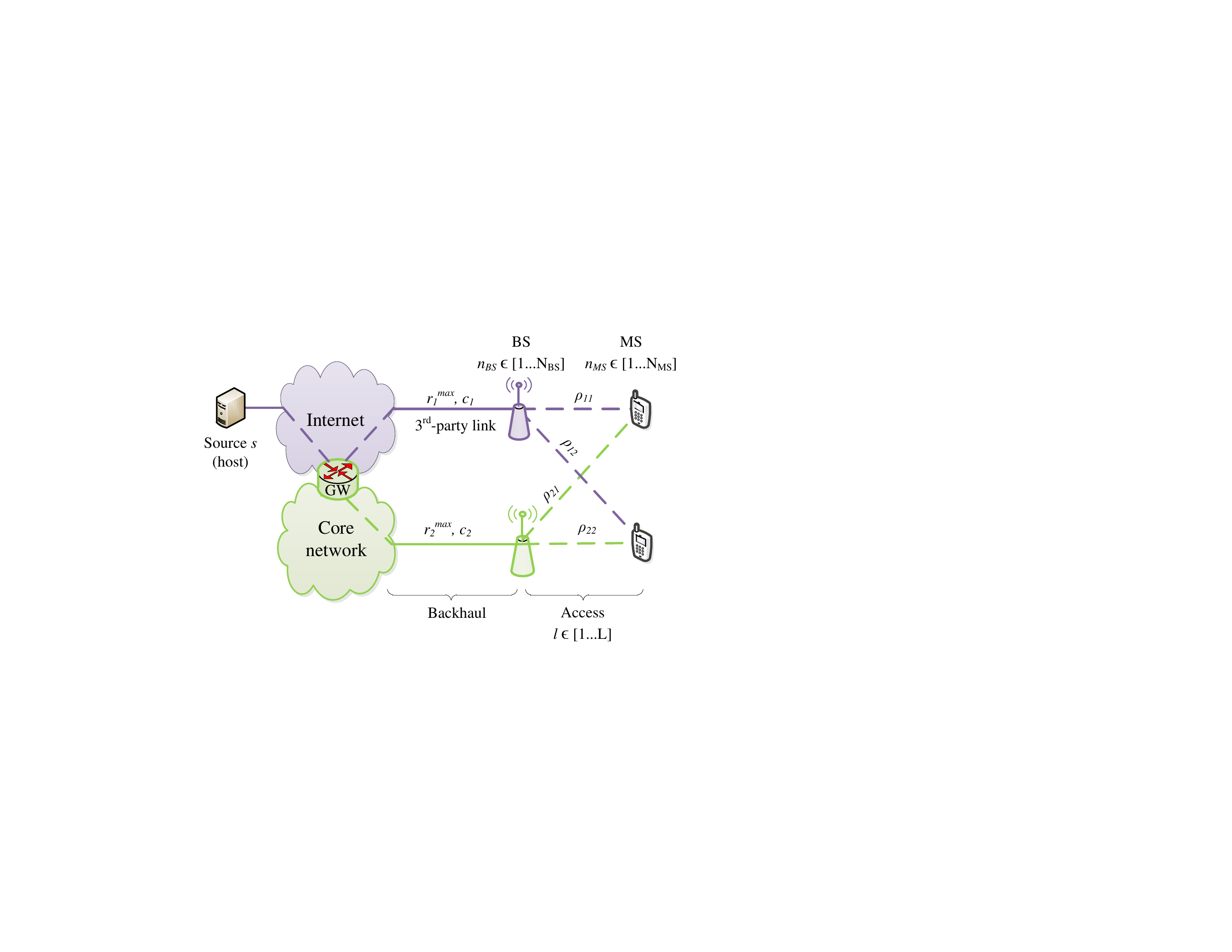}
\caption{Heterogeneous network architecture where operator-controlled
cells (drawn in green) are combined with third-party controlled femtocells
(drawn in purple).   }
\label{fig:cell_selection_net}
\end{figure}
\end{center}


As described in the Introduction, we propose that third parties use
open-access femtocells to provide service to the mobile subscribers
of a cellular operator.
The basic network architecture is shown in
Fig.~\ref{fig:cell_selection_net} which follows the standard
model of 3GPP LTE/SAE heterogeneous networks
\cite{OlssonSRFM:09,ZhangRoche:10}.
As shown, there are two classes of base station
cells in the proposed model: \emph{operator-controlled}
and \emph{third-party}.
The operator-controlled cells (set $BS_M$) are the standard
BS nodes connected directly to the operator core network and
managed by the operator.  These are typically macro- or microcells, hence
the subscript $M$.  The third party cells (set $BS_F$),
are the open access femtocell BS nodes installed by a third party and
connected to the third-party ISP network, i.e.\ the Internet.

In the proposed model, the third-party and operator first
agree on some access pricing,
perhaps a cost per unit time or unit data that the operator will reimburse
the third-party provider when mobile subscribers connect to
their network --
we will see that our optimization methodology can incorporate a range of
pricing models.
Then, at any time, the network can choose to assign mobiles to either
operator-controlled or third-party cells
depending on the channel and interference conditions,
traffic loading and access pricing.

There are several convenient features of this deployment model:
\begin{itemize}
\item \emph{Deployment ease and scalability:}
Most importantly, the use of third party femtocells offers
a scalable and cost effective approach to increase capacity:
Since femtocells are low-cost, plug-and-play devices,
third parties can install these devices themselves, thereby immediately
creating an abundance of cell sites virtually for free.  In addition,
as described in the Neighborhood Small Cell concept \cite{Qualcomm-NSC},
when the cellular operator also provides a broadband
residential and enterprise ISP service
(such as Verizon's FiOS or AT\&T's U-Verse service),
they could include a femtocell within the WiFi access point
given to the subscriber, thereby automatically enabling femtocell capabilities
in every subsciber's location.

\item \emph{Direct operator to third party economic relationship:}
The economic interaction with respect to
access pricing can be made entirely between the cellular operator and the
third-party provider -- the mobile subscriber need not be involved.
This arrangement is possible since, in networks such as 3GPP LTE,
for mobiles in connected mode, the network makes all
decisions on which cells serve the mobiles \cite{OlssonSRFM:09}.
Thus, the network and third-party
provider can come to an agreement on access pricing, and then the network
can decide, in each time instant, whether to have its mobiles served by
the third party cells based on channel conditions, network load,
and other factors.

\item \emph{Transparency to mobiles:}
Related to the above point is that,
from the mobile station perspective, both classes of base station cells
-- third-party and operator-controlled -- are
completely identical, potentially using the same radio access technology
and potentially operating in the same spectrum band.\footnote{Our interference
model presented in the following section addresses co-channel and separate channel deployments.}
Thus, the third party vs.\ operator ownership
is transparent to the mobiles.  This feature is beneficial since
the mobile user is really concerned with the quality of the connectivity
and has no interest \emph{per se} in who provides that connection.

\item \emph{Correctly matched incentives:} Since the operator will only reimburse
third parties when Open Subscriber Group (OSG) mobiles are served
 by the third party cell, the operator does not need to enforce proper
installment and operation of the femtocell.  Third parties will be naturally
incentivized to keep their femtocell on and well-positioned to attract
the operator to move its mobile subscribers onto its network to receive
payment.  On the other hand, if the third-party network is busy and cannot
support the additional mobile traffic,
the third-party is free to shut off or
cap the rate to the femtocell and the cellular
operator can adjust its cell selection decisions accordingly.
As we will see in the next section, our user association algorithm
can account for backhaul limits on the third party cells.

\item \emph{Minimal changes to existing standards:}
All the hooks necessary for the proposed third-party offload can be
handled within the existing cellular standards.
For example, in 3GPP LTE,
the mobiles already provide the network with all the measurement reports
necessary to determine the airlink conditions and received signal
strengths to make the handover decisions.
Also, in the current LTE network model,
all traffic from the public Internet is routed first through a
gateway\footnote{The gateway, in this case,
serves the combined function of the LTE P-GW/S-GW.
We consider these to be co-located nodes.}
before being tunneled to the base stations, whether the
base stations cells are operator-controlled or operated by a third-party
outside the operator core network \cite{OlssonSRFM:09}.
\footnote{Note that while protocols such as Selected IP Traffic Offload (SIPTO)~\cite{3GPP23.829}
allow small cell traffic to bypass the CN and be offloaded directly to the Internet,
we do not consider this case here since network and transport layer mobility
functions need to be handled by the CN in any case.}
Thus, the network can both
monitor the exact amount of time and data to each cell type for charging
and to measure link quality.

\end{itemize}

\section{User Association Optimization}
\label{sec:algorithm}

As mentioned above, a key technical challenge in realizing the proposed
architecture is that the cellular provider requires a good algorithm
for \emph{user association}:  the cellular provider must determine,
in each time instant, how to assign
mobile users between the third-party and operator-controlled cells while
accounting for channel and interference conditions, access prices and loading.
To this end, we use a utility maximizing algorithm
proposed in \cite{KimFQR:13-arxiv} and
\cite{YeRon:12} for optimized user association
in heterogeneous networks, incorporating access price into the utility function.

\subsection{Optimization Formulation}

Returning to Fig.~\ref{fig:cell_selection_net}, let MS$i$, $i=1,\ldots N_{MS}$
denote the set of mobile stations and BS$j$, $j=1,\ldots, N_{BS}$
set of base station cells,
the latter set including both operator-controlled and third-party cells.
For each MS$i$, let $\Gamma_{\rm MS}(i)$
denote the indices $j$ such that BS$j$ can potentially serve the mobile.
Similarly,
let $\Gamma_{\rm BS}(j)$ be the set of mobiles potentially served by BS$j$.
Also, let $r_{ij}$ and $w_{ij}$ be the rate and bandwidth allocated
to MS$i$ from BS$j$ for $j \in \Gamma_{\rm MS}(i)$.
Let $\rbar^{MS}_i$ and $\rbar^{BS}_j$
be the total rate to the MS$i$ and BS$j$ respectively,
which must satisfy the constraints
\beq \label{eq:rbar}
    \rbar^{MS}_i \leq \sum_{j \in \Gamma_{\rm MS}(i)} r_{ij}, \quad
    \rbar^{BS}_j \geq \sum_{j \in \Gamma_{\rm BS}(j)} r_{ij}.
\eeq
Absent of carrier aggregation~\cite{YuanZWY:10},
mobiles in LTE are typically only served by one cell at a time.
If we let $\rbf$ be the vector of the rates $r_{ij}$, we will denote
this single path constraint as
\beq \label{eq:singPath}
    \rbf \in \Sbf := \left\{
        \rbf~:~ r_{ij}=0 \mbox{ for all but one } j \in \Gamma_{\rm MS}(i)\right\}.
\eeq
Now, following \cite{KimFQR:13-arxiv}, we attempt to find rates to maximize
some utility function of the form
\beq \label{eq:util}
    U(\rbfbar^{MS}) = \sum_{i=1}^{N_{MS}} U_i(\rbar^{MS}_i),
\eeq
for some utility functions $U_i(\rbar^{MS}_i)$.
To account for the backhaul costs that the operator must pay
the third party providers, we assume there is a cost of the form
\beq \label{eq:cost}
    C(\rbfbar^{BS}) = \sum_{j=1}^{N_{BS}} C_j(\rbar^{BS}_j),
\eeq
where $C_j(\rbar^{BS}_j)$ is the cost for the traffic on BS$j$
that the operator will have to pay the third-party that owns BS$j$.
BS nodes belonging to $BS_M$ represent operator-deployed cells connected over statically-provisioned backhaul links will generally have
zero cost ($C_j(\rbar^{BS}_j)=0$),
since we consider existing infrastructure to be a sunk cost for the operator.
Nodes in $BS_F$ are third-party cells
and have some positive cost ($C_j\rbar^{BS}_j)>0$) that they charge the operator.
The cost functions $C_j(\rbar^{BS}_j)$ can also be used to incorporate
rate limits on either the third-party or operator-deployed cells
due to finite backhaul capacity on those cell sites.

The goal is to maximize a \emph{net utility}, $U(\rbfbar^{MS}) - C(\rbfbar^{BS})$,
subject to constraints on the rates.  The rate limits depend on the
channel and interference conditions, which we model using a linear mixing
interference model proposed in \cite{RanganM:12}.
Let $z_{ij}$ be the interference power on the link from BS$j$ to MS$i$.
As described in \cite{KimFQR:13-arxiv}, assuming that
the base stations radiate a fixed power per unit bandwidth,
the vector of interference powers must satisfy a constraint of the form
\beq \label{eq:zGw}
    \zbf \geq \Gbf\wbf,
\eeq
where $\wbf$ is the vector of bandwidth allocations and $\Gbf$ is an appropriate
gain matrix.
The rate on the links must then satisfy
\beq \label{eq:rateij}
    r_{ij} \leq w_{ij}\rho_{ij}(z_{ij}),
\eeq
where $\rho_{ij}(z_{ij})$ is the spectral efficiency (rate per unit bandwidth)
as a function of the interference level $z_{ij}$.  Also, the bandwidths must satisfy
some constraints of the form
\beq \label{eq:bwcon}
    \sum_{i \in \Gamma_{\rm BS}(i)} w_{ij} \leq \overline{w}_i,
\eeq
where $\overline{w}_i$ is the total bandwidth available in BS$j$.
By appropriate selection of the gain matrix $\Gbf$,
this formulation can incorporate both co-channel deployments of the third-party
and operator-controlled cells where the two types of cells use the same
bandwidth and interfere with one another, or separate channel deployments
where there is no interference.

Now let $\thetabf$ be the vector of all the unknowns
\beq \label{eq:theta}
    \thetabf = (\rbf,\rbfbar^{BS}, \rbfbar^{MS},\wbf,\zbf)^T.
\eeq
The constraints \eqref{eq:rbar}, \eqref{eq:zGw} and \eqref{eq:bwcon}
can be replaced by inequality constraints represented in a matrix form
\beq \label{eq:thetaIneq}
    \Abf\thetabf \leq \bbf,
\eeq
for appropriate choice of the matrix $\Abf$ and vector $\bbf$.
We can then write the user association problem as the optimization
\begin{subequations} \label{eq:optCon}
\beqa
    \lefteqn{ \max_{\thetabf} U(\rbfbar^{MS})-C(\rbfbar^{BS}) }
       \label{eq:opt}  \\
    &s.t.&     \Abf\thetabf \leq \bbf  \label{eq:AtbCon} \\
    & & \rbf \leq \wbf \cdot \rho(\zbf) \label{eq:rwzCon} \\
    & & \rbf \in \Sbf,  \label{eq:rsCon}
\eeqa
\end{subequations}
where the constraint \eqref{eq:rwzCon} is a vector shorthand for the constraints
\eqref{eq:rateij}.

\subsection{Dual Decomposition Algorithm}

The optimization \eqref{eq:optCon} is, in general, non-convex
due to the nonlinearity of the
interference rate-interference constraints \eqref{eq:rwzCon} and the single
path constraints \eqref{eq:singPath}.
However, following \cite{KimFQR:13-arxiv}, we can find an approximate solution
to this problem in two steps.  First, we initially ignore the single path
constraint~\eqref{eq:singPath}.  The resulting optimization will produce a vector
with rates on all paths to the mobiles.  We then simply ``truncate" the solution
to take the path with the largest rate.  This truncation procedure is well-known
in networking theory \cite{PioroMedhi:04} and often introduces little error,
since the multipath solution tends to concentrate on dominant single paths.

Unfortunately, even with the multipath approximation, the optimization
\eqref{eq:optCon} will be non-convex.
However, as in \cite{KimFQR:13-arxiv},
we can show that the optimization admits a separable dual decomposition.
Specifically, the Lagrangian corresponding to the optimization
\eqref{eq:optCon} without the single path constraint \eqref{eq:rsCon}
is given by
\beq \label{eq:Lag}
   L(\thetabf, \mubf) := U(\rbfbar^{MS})  - C(\rbfbar^{BS}) -
    \mubf^Tg(\thetabf),
\eeq
where $g(\thetabf)$ is the vector of constraints
\beq \label{eq:gtheta}
    g(\thetabf) = \Bigl( \Abf\thetabf - \bbf,
                     \rbf - \wbf \cdot \rho(\zbf) \Bigr)^T,
\eeq
and $\mubf \geq 0$ is the vector of dual parameters partitioned
conformably with $g(\thetabf)$
\beq \label{eq:mu}
    \mubf = (\mubf^\theta, \mubf^r)^T.
\eeq

Now central to using dual optimization methods is the ability to
compute the maxima of the form
\beq \label{eq:thetahat}
    \thetabfhat(\mubf)
        := \argmax_{\thetabf}  L(\thetabf,\mubf) - \Phi(\thetabf),
\eeq
for some augmenting function $\Phi(\thetabf)$.  Using a similar argument
as in \cite{KimFQR:13-arxiv}, the following lemma shows that,
under the assumption of
a separable augmenting function $\Phi(\cdot)$, the
optimization \eqref{eq:thetahat} admits a separable dual decomposition.

\begin{lemma} \label{lem:LagSep}
Let $\Phi(\thetabf)$ be any separable function of the form
\beqa
    \lefteqn{ \Phi(\thetabf) = \Phi(\rbf,\rbfbar^{BS}, \rbfbar^{MS},\wbf,\zbf) }
    \nonumber \\
    &=& \sum_{ij} \Bigl[ \phi^r_{ij}(r_{ij}) + \phi^z_{ij}(w_{ij},z_{ij})\Bigr]
    \nonumber \\
    & & + \sum_j \phi^{BS}_j(\rbar^{BS}_j) +
    \sum_i \phi^{MS}_i(\rbar^{MS}_i)
     \label{eq:barrier}
\eeqa
for some functions $\phi^r_{ij}(\cdot)$, $\phi^{BS}_j(\cdot)$,
$\phi^{MS}_i(\cdot)$ and $\phi^{z}_{ij}(\cdot)$.
Let $\mubf$ be any vector of Lagrange parameters and
let $\thetabfhat(\mubf)$ be the corresponding maxima for
the dual optimization augmented by $\Phi(\thetabf)$, namely
\beqa
    \lefteqn{  \thetabfhat(\mubf) =
    (\rbfhat, \rbfhat^{BS},
    \rbfhat^{MS}, \wbfhat, \zbfhat) } \nonumber \\
    &:=& \argmax_{\thetabf} \Bigl[ L(\thetabf,\mubf)
        -  \Phi(\thetabf) \Bigr]. \hspace{2cm}
        \label{eq:thetahatsol}
\eeqa
Then, the components of $\thetabfhat(\mubf)$
are given by the solutions to the optimizations
\begin{subequations}
\beqan
    \rhat_{ij} &\in& \argmin_{r_{ij}} \Bigl[ \phi^r_{ij}(r_{ij}) +
            (\lambda^r_{ij}+\mu^r_{ij})r_{ij} \Bigr] \label{eq:rhatmu} \\
    \rhat^{BS}_{j} &\in& \argmin_{\rbar^{BS}_{j}}
        \Bigl[ C_j(\rbar^{BS}_j) + \phi^{BS}_j(\rbar^{BS}_j)  +
            \lambda^{BS}_j\rbar^{BS}_j \Bigr] \label{eq:rbsmu} \\
    \rhat^{MS}_{i} &\in& \argmax_{\rbar^{MS}_{i}}
        \Bigl[ U_i(\rbar^{MS}_i) - \phi^{MS}_i(\rbar^{MS}_i) -
            \lambda^{MS}_i\rbar^{MS}_i \Bigr] \label{eq:rmsmu}
\eeqan
and
\beqan
    \lefteqn{  (\what_{ij},\zhat_{ij})
    \in \argmax_{w_{ij},z_{ij}} \Bigl[ \mu^r_{ij}w_{ij}\rho_{ij}(z_{ij}) } \nonumber \\
     & & - \lambda^z_{ij}z_{ij}-\lambda^w_{ij}w_{ij}
     - \phi^z_{ij}(w_{ij},z_{ij}) \Bigr], \label{eq:wzhatmu}
\eeqan
\end{subequations}
where the dual parameters $\mubf$ are partitioned as in \eqref{eq:mu}
and the parameters $\lambda$ are the components
\[
    \lambdabf = (\lambdabf^r,\lambdabf^{BS},\lambdabf^{MS},\lambdabf^{w},
    \lambdabf^{z}) := \Abf^T\mubf^\theta,
\]
where the vector $\lambdabf$ has been partitioned conformably with $\thetabf$
in \eqref{eq:theta}.
\end{lemma}
\begin{proof}
This result follows immediately from the
separable structure of the objective function and constraints.
\end{proof}

\medskip
The lemma shows that the vector optimization \eqref{eq:thetahatsol}
separates into a set of simple one and two-dimensional optimizations
over the components of $\thetabf$ and can thus be computed easily
for any dual parameters $\lambdabf$.
As discussed in \cite{KimFQR:13-arxiv},
this separability property has two key consequences that follow immediately
from standard optimization theory \cite{NorWright:06}:
First, one can use weak duality to efficiently to compute an
upper bound on the maximum net utility.  Specifically,
the net utility is bounded by
\beq \label{eq:dualBnd}
    \max_{\thetabf} U(\thetabf) - C(\thetabf) \leq
        \max_{\thetabf} L(\thetabf,\mubf),
\eeq
for any dual parameters $\mubf \geq 0$.  Moreover, the right-hand side
of \eqref{eq:dualBnd} is convex in $\mubf$, and since the dual maxima
is separable, one can compute its value and gradient for any $\mubf$.
Thus, one can efficiently minimize the right hand side of \eqref{eq:dualBnd},
providing a computable upper bound on the net utility.
Although the dual upper bound may not
be tight (i.e.\ the optimization may not be strongly dual),
the bound provides a computable metric under which any
approximate algorithm can be compared against.

A second consequence of a computable dual maxima is that
one can efficiently implement several well-known augmented Lagrangian techniques
to find approximate maxima to the net utility $U(\thetabf) - C(\thetabf)$.
These methods include various inexact versions of alternating direction method
of multiplier methods \cite{BoydPCPE:09,Zhang:JSC:11} --
see \cite{KimFQR:13-arxiv} for more details.

A valuable property is that this dual decomposition applies
in very general circumstances.  In particular, one can consider arbitrary
cost and utility functions, and SNR-to-rate mappings $\rho_{ij}(\cdot)$.
Thus, the framework provides a tractable and general methodology for a large
class of networks, interference scenarios and pricing schemes with computable
upper bounds on performance.

\section{System Evaluation}

\subsection{Evaluation Methodology}
\label{sec:simulation_setup}

To evaluate the opportunistic backhaul model, we conducted a network simulation
using a simplified version of industry-standard
models for evaluating multi-tier networks \cite{3GPP36.814,FemtoForum:10,ITU-M.2134}.
Our goal was to determine the gain in network capacity
as a function of the number
of third party providers and offloaded backhaul capacity.

Two deployment models were considered:  \emph{real-world}
and \emph{stochastic}.  In the real-world model, the
operator-controlled BS locations
were based on actual cell sites
as reported in the OpenCellId database, which has logged data for over 500,000
individual cells in the US \cite{opencellid}.
For the third-party femtocells,
following \cite{Qualcomm-NSC}, we assume that the third party femtocells can be
co-located with current residential and enterprise WiFi access points
(APs) -- the theory being that the owners of these WiFi
APs would be the potential third parties to offer connectivity to
the mobile subscribers.  The WiFi AP locations were estimated from the
freely-available Wigle.net database.
Contributors to the Wigle project have compiled over 88 million unique observations of WiFi stations across the globe \cite{wigle}.

Our simulations considered two test locations:
a 1km$^2$ area of the East Village in Manhattan, New York City
-- representative of a dense urban deployment -- and a similar-sized area in
Passaic, NJ as representative of a typical suburban deployment.
Fig.~\ref{fig:east_village_map} is an aerial map of WiFi APs and
cell sites in the selected area of Manhattan
as reported in the OpenCellId and Wigle.net databases,
plotted together in Google Earth\texttrademark \cite{google_earth}.
For the observed population of WiFi APs in the 1km$^2$ areas
of NYC and Passaic, NJ,
we uniformly sample a percentage of these nodes in the domain
$\{2\%, 5\%, 10\%, 15\%, 20\%\}$, which represents the
adoption rate of third-party WiFi owners
that agree to provide backhaul service to the cellular operator
through a femtocell co-located at a WiFi AP.

The equivalent number of nodes for both the urban
and suburban locations is given in Table~\ref{table:scenarios}.
We immediately see that, according to the databases we sampled,
there are more than 2000 WiFi APs for each
cellular microcell in Manhattan.  This vast number of WiFi APs suggests
that if, even a small fraction of open-access
femtocells can be co-located at current WiFi locations,
the cellular capacity can be massively increased.

\begin{figure}
\centering
\includegraphics[width=2.8in]{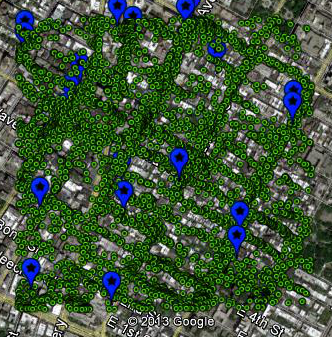}
\caption{Google Earth\texttrademark map of reported microcell locations
(blue balloons) and a fraction of the WiFi sites (green circles)
in the East Village, Manhattan.
Since the WiFi access points vastly outnumber the current cellular base stations
(more that 2000 to 1 in this case), cellular capacity can be significantly
increased if femtocells can be co-located at even a small fraction
of the WiFi sites.  }
\label{fig:east_village_map}
\end{figure}

\begin{table} \footnotesize
\caption{Number of operator-controlled and third-party cells
and number of UEs in the urban and suburban test cases}
\label{table:scenarios}
\centering
\begin{tabular}{ |p{0.2\columnwidth}|p{0.1\columnwidth}|p{0.53\columnwidth}| }
\hline
\textbf{Environment} & \textbf{Node Type} & \textbf{Number of Nodes} \\ \hline
\multirow{3}{*}{\parbox{0.2\columnwidth}{Urban \\ micro-only and \\ micro+femto }} & Micro & $N_M = 17$ \\ \cline{2-3}
& Femto & $N_F\in\{722,1805,3610,5415,7220\}^*$ \\ \cline{2-3}
& UE & $N_{MS} = 425$ (25/micro) \\ \hline
\multirow{3}{*}{\parbox{0.2\columnwidth}{Suburban: \\ macro-only and \\
macro+femto}} & Macro & $N_M = 9$ \\ \cline{2-3}
& Femto & $N_F\in\{66,165,330,495,660,825\}^*$\\ \cline{2-3}
& UE & $N_{MS} = 225$ (25/macro) \\ \hline
\end{tabular}
\linebreak
\linebreak
\footnotesize{*$N_F$ computed from the adoption rate: $\{2\%, 5\%, 10\%, 15\%, 20\%\}$ of observed WiFi APs}
\end{table}


For completeness, we compare our results based on real-world
locations with industry-standard \emph{stochastic models}
\cite{3GPP36.814,ITU-M.2134} used widely in evaluating cellular systems.
For the urban scenario in both the real-world and stochastic models,
we take the operator-deployed nodes to be microcells, which are omnidirectional transmitters with power, bandwidth and other parameters coinciding with the
3GPP urban microcell model in \cite{3GPP36.814,ITU-M.2134}. Femtocells are similarly configured based on these parameters and are subject to a maximum backhaul rate in the set $\{10,20,30,40,50\}$ Mbps.\footnote{Backhaul rate constraint values are based on  broadband services commonly offered by many ISPs in the US.} These values for the rate constraint parameter are assigned uniformly over the set of femtocells. For the suburban scenario (for both real-world and stochastic models),
we consider operator BSs to be three-way-sectorized macrocells.
For the stochastic models, the densities of the
operator-controlled macro/microcells and third-party femtocells
were adjusted to match the densities observed in the real-world data.
Other salient parameters are given in Table~\ref{table:network_params}.
Note that the femtocells have an additional 20 dB of path loss to
account for the wall loss assuming the femtocell is deployed indoors.

For each of the scenarios (urban / suburban and real-world / stochastic),
we follow the standard evaluation methodology in \cite{3GPP36.814,ITU-M.2134}
to assess the downlink capacity.  Specifically, we generate 10 random
instances of each of the networks;  each instance is called a \emph{drop}.
In each random drop, we run the optimization described in Section~\ref{sec:algorithm} to determine the optimal user association
between the third-party and operator-controlled cells.
In reality, we envision that
this optimization would be conducted in the operator's network.

In conducting the optimization, we assume a similar model
as \cite{MogEtAl:07}, where the spectral efficiency
$\rho_{ij}(z_{ij})$ in \eqref{eq:rateij} is given by the Shannon capacity
with a loss of 3~dB
with a maximum value of 4.8 bps/Hz
(corresponding to 64-QAM at rate 5/6).
For the utility \eqref{eq:util},
we assume a proportional fair metric $U_i(\rbar^{MS}_i) = \log(\rbar^{MS}_i)$.
The backhaul capacity is assumed to be zero for operator-controlled cells,
while third-party cells charge a linear cost so that
\beq \label{eq:Cfemto}
    C_j(\rbar^{BS}_j) = \begin{cases}
        p\rbar^{BS}_j, & j \in BS_F \mbox{ (i.e. femtocell)} \\
        0, & j \in BS_M \mbox{ (i.e. macro/microcell)}
        \end{cases}
\eeq
where $p$ represents the
cost per unit backhaul capacity in the femtocell
relative to the utility.
The price $p$ will be varied.

\begin{center}
\begin{table}[h]\footnotesize
\centering
\caption{Network model parameters}
\label{table:network_params}
\begin{tabular}{ |p{0.01\columnwidth}|p{0.24\columnwidth}|p{0.53\columnwidth}| }
\hline
& \textbf{Parameter} & \textbf{Value} \\ \hline
\multirow{6}{*}{\rotatebox[origin=c]{90}{Urban microcell~~~~~}} & Topology & Uniform with wrap-around \\ \cline{2-3}
& Total TX power & 30 dBm \\ \cline{2-3} 
& BW & 10 MHz (FDD DL) \\ \cline{2-3} 
& Antenna pattern & omni \\ \cline{2-3} 
& Micro $\leftrightarrow$ UE path loss & $15.3 + 37.6\log_{10}(R)$ ($R$ in km) \\ \cline{2-3} 
& Micro $\leftrightarrow$ UE lognormal shadowing & 10dB std. dev; 50\% inter-site correlation; 100\%
intra-site correlation \\ 
\hline

\multirow{7}{*}{\rotatebox[origin=c]{90}{Femtocell~~~~~~~}} & Topology & Uniform with wrap-around \\ \cline{2-3}
& Total TX power & 20 dBm \\ \cline{2-3}
& BW & 10 MHz (FDD DL) \\ \cline{2-3} 
& Antenna pattern & omni \\ \cline{2-3}
& Femto $\leftrightarrow$ UE path loss & $15.3 + 37.6\log_{10}(R) + 20dB$ ($R$ in km) \\ \cline{2-3}
& Femto $\leftrightarrow$ UE lognormal shadowing & 8dB std. dev \\ \cline{2-3}
& Backhaul max rate & $r_{max} \in \{10,20,30,40,50\}$ Mbps \newline
uniformly assigned to nodes\\
\hline

\multirow{6}{*}{\rotatebox[origin=c]{90}{Suburban macrocell~~~~~~}}
& Topology & Hexagonal (3-way sectorized) with wrap-around  \\ \cline{2-3}
& Total TX power & 46 dBm \\ \cline{2-3}
& BW & 10 MHz (FDD DL) \\ \cline{2-3} 
& Antenna pattern & $ A(\theta) = -min\{12 (\frac{\theta}{\theta_{3dB}})^2,A_m \} $

\\ \cline{2-3}
& Macro $\leftrightarrow$ UE path loss & $15.3 + 37.6log_{10}(R)$ ($R$ in km) \\ \cline{2-3} 
& Macro $\leftrightarrow$ UE lognormal shadowing & 8dB std. dev\\ 
\hline

\multirow{5}{*}{\rotatebox[origin=l]{90}{Global~~~~~~~~~~~~~}}
& Carrier frequency & 2.1 GHz \\ \cline{2-3}
& Total area & 1000x1000m \\ \cline{2-3}
& UE distribution & Uniform, 25 per macro/microcell \\ \cline{2-3}
& Mobility & constant \\ \cline{2-3}
& Traffic model & full buffer \\ \cline{2-3} 
& Fading & none \\ \cline{2-3}
& Link capacity &
$\scriptstyle C = W*min( \log(1+10^{-\frac{1}{10}\beta} SNR), \rho_{max})$

$\beta = 3dB$ (loss from Shannon cap.*)

$\rho_{max}$ (max. spectral efficiency)\\
\hline
\end{tabular}
\newline
\newline
\parbox{0.8\columnwidth}{\footnotesize{*We determine the acheivable rate based on the loss from Shannon capacity, as discussed in \cite{MogEtAl:07}.}}
\end{table}
\end{center}

\subsection{Potential Capacity Gain with Offloading}
\label{sec:capGainMax}

To estimate the maximum possible capacity gain with third-party
offload, we first
consider the case where the third-party providers lease their capacity
at zero cost (i.e.\ $p=0$ in \eqref{eq:Cfemto}).
Fig.~\ref{fig:rate_gain} shows that the gain
in mean throughput per mobile
as a function of the adoption rate -- the fraction of WiFi AP locations
where open-access femtocells are co-located.
We see that the total capacity can be increased significantly.  For example,
a 5\% adoption yields more than 20x increase in user throughput in the
real-world urban setting and a 6x increase in the real-world suburban model.
The maximum gain in the suburban model is not as high as the urban
 setting since the density of third-party cells is lower.
 Also, the maximum gains in both cases begin to saturate since we fix the
 number of UEs per marcros.  Therefore, adding more femtocells
 eventually has little value -- a phenomena also observed in
 \cite{Qualcomm-NSC}.  The finite backhaul rates on the femtocells
 also limits the gain.

\begin{figure}
\begin{subfigure}[t]{1.0\linewidth}
\centering
\includegraphics[width=0.8\columnwidth, trim=1.3in 3.4in 1.3in 3.4in]{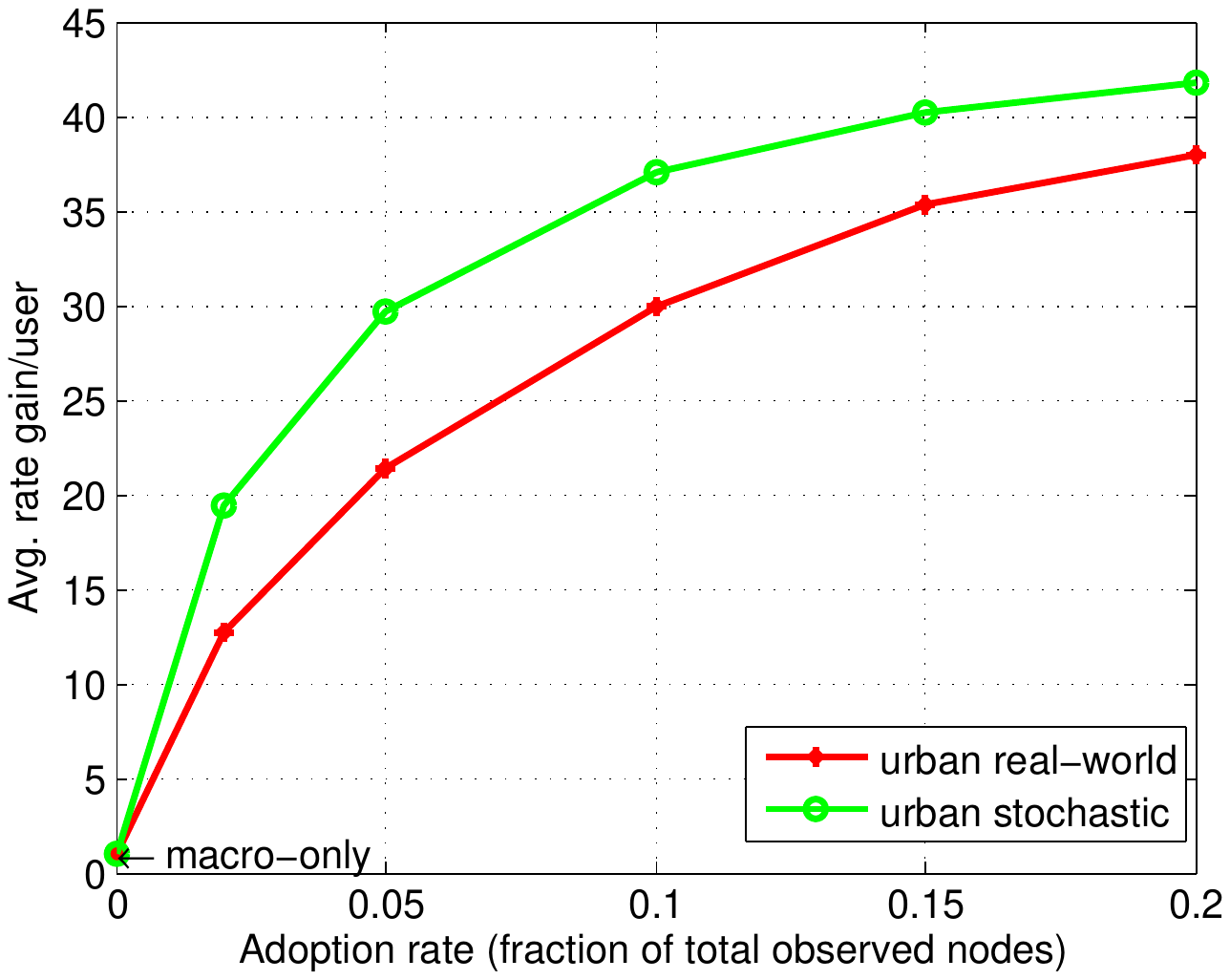}
\caption{Urban scenario}
\label{fig:rate_gain_urban}
\end{subfigure}

\begin{subfigure}[t]{1.0\linewidth}
\centering
\includegraphics[width=0.8\columnwidth, trim=1.3in 3.4in 1.3in 3.4in]{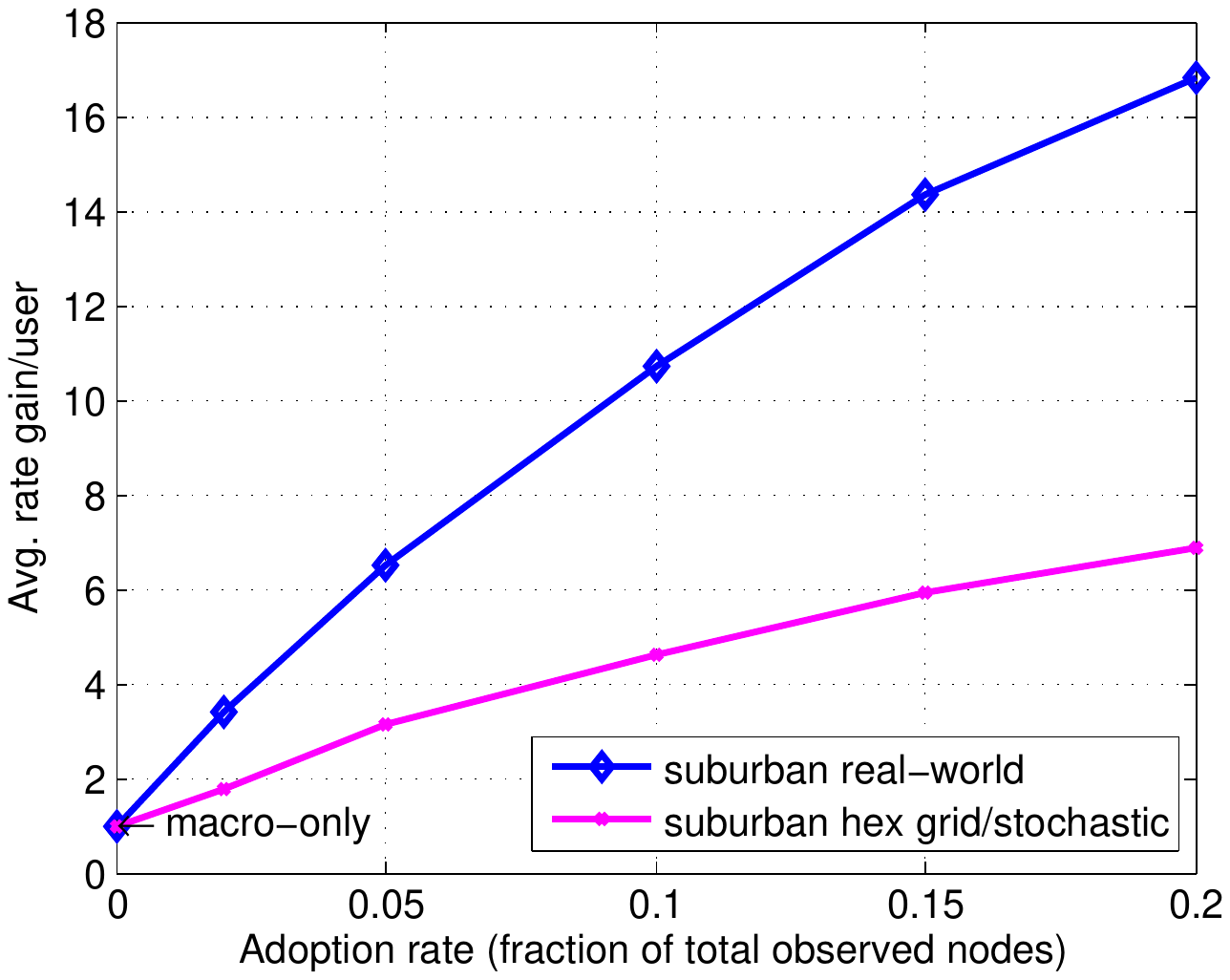}
\caption{Suburban scenario}
\label{fig:rate_gain_suburban}
\end{subfigure}

\caption{Average rate gain as a function of the adoption rate
in urban and suburban deployments.  We see that
 if even a small fraction (5\%) of WiFi owners agree to install open-access
 femtocells, cellular capacity can increase more than 20 fold
 in a dense urban environment and a factor of 6 in a suburban deployment.
 \label{fig:rate_gain} }

\end{figure}

\begin{figure}
\begin{subfigure}[t]{1.0\linewidth}
\centering
\includegraphics[width=0.8\columnwidth, trim=1.3in 3.4in 1.3in 3.4in]{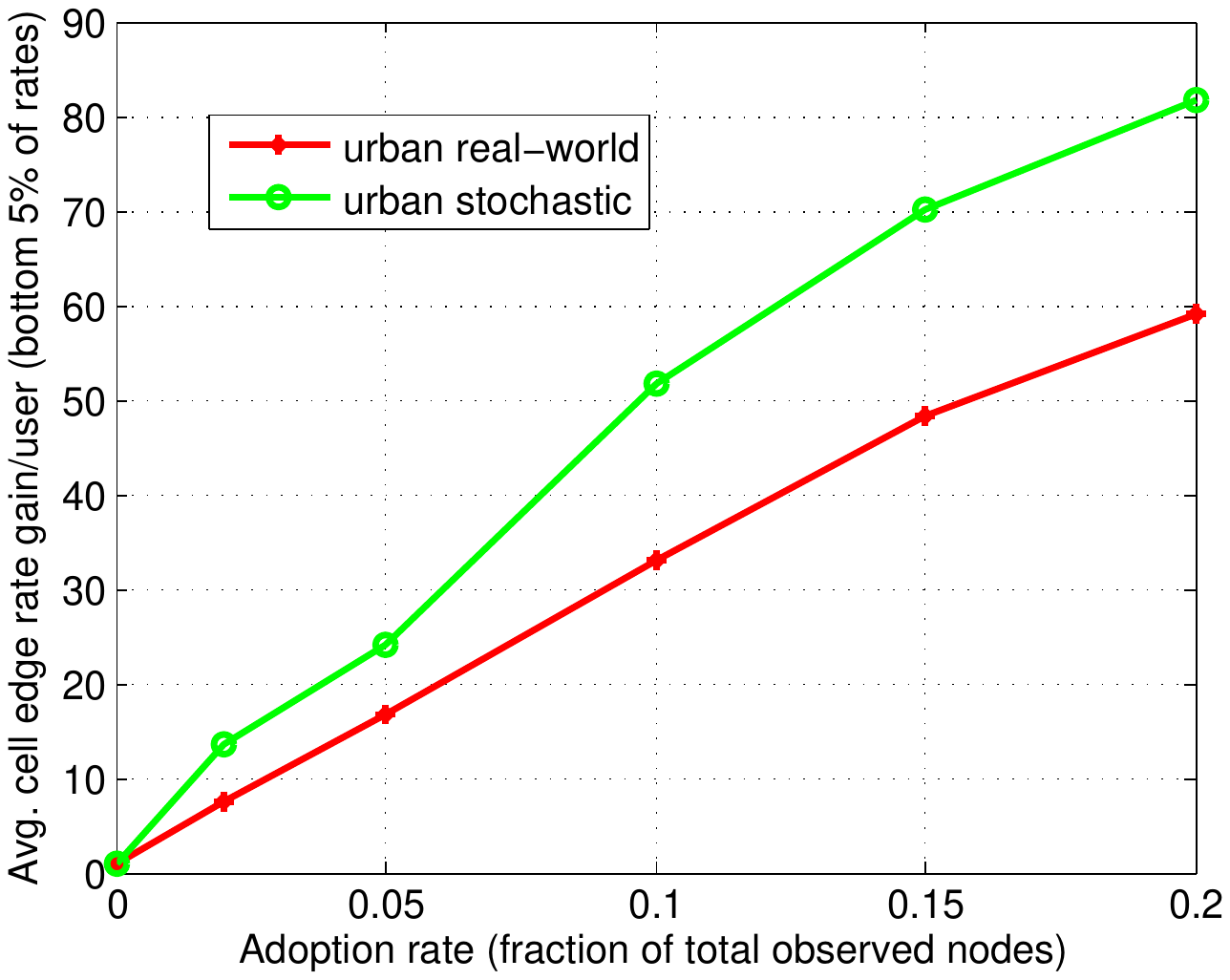}
\caption{Urban scenario}
\label{fig:edge_gain_urban}
\end{subfigure}

\begin{subfigure}[t]{1.0\linewidth}
\centering
\includegraphics[width=0.8\columnwidth, trim=1.3in 3.4in 1.3in 3.4in]{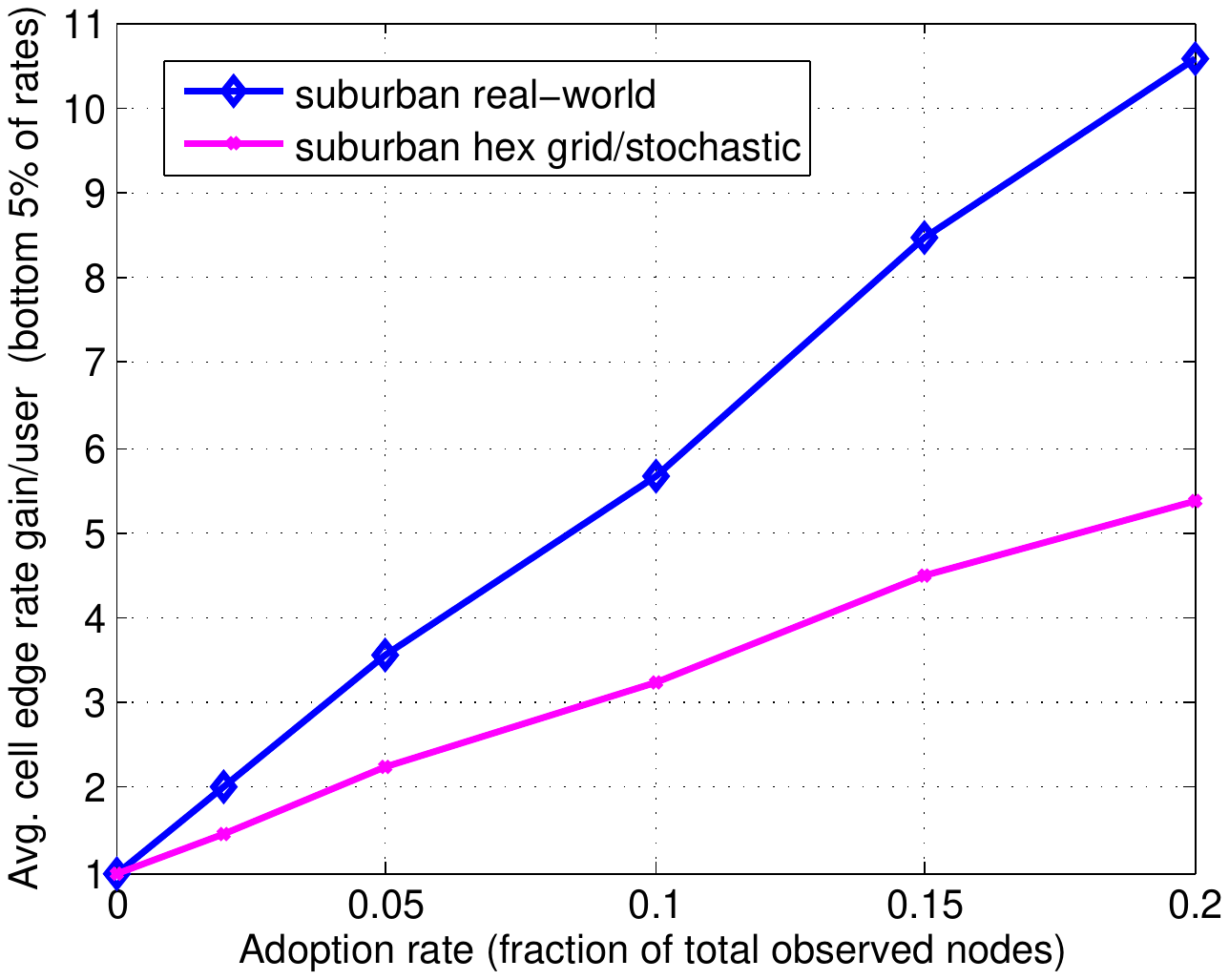}
\caption{Suburban scenario}
\label{fig:edge_gain_suburban}
\end{subfigure}

\caption{Average 5\% cell-edge rate gain as a function of the adoption rate
in urban and suburban deployments.
 \label{fig:cell_edge_gain} }
\end{figure}

Fig.~\ref{fig:cell_edge_gain} similarly plots the increase in
the cell-edge throughput.  As defined in \cite{3GPP36.814},
the cell-edge rate is the rate of the 5\% percentile UE
in each drop.  We see similar gains at the cell edge as the mean gains,
suggesting that the gains are uniformly experienced across the cell.
Table~\ref{table:gains} states the gain value for the urban case
at the 5\% adoption.

\begin{center}
\begin{table}\footnotesize
\caption{Capacity and cell edge gains with 5\% femtocell adoption
based on the urban real-world model.}
\label{table:gains}
\centering
\begin{tabular}{ |>{\raggedright}p{0.22\columnwidth}|p{0.18\columnwidth}|p{0.18\columnwidth}|p{0.1\columnwidth}|}
\hline
\textbf{Scenario} & \textbf{Macro-only} & \textbf{Macro+Femto} & \textbf{Gain} \\ \hline
Avg. UE rate (Mbps) & 0.63 & 13.38 & 21.39 \tabularnewline \hline
Avg. 5\% cell edge rate (Mbps) & 0.08 & 1.36 & 16.71 \tabularnewline \hline
Utility (geometric mean rate, Mbps) & 0.41 & 9.35 & 22.66 \tabularnewline \hline
\end{tabular}
\end{table}
\end{center}
%

\subsection{Adding Third-Party Pricing}
\label{sec:simPricing}

The results in the previous subsection assumed that
the relative cost of the third-party backhaul
(the variable $p$ in \eqref{eq:Cfemto}) was zero.
Of course, in reality, third parties will not generally offer backhaul
for free, so we need to consider the capacity gain
with non-zero pricing.
Unfortunately, there is no way to directly determine the ``correct"
relative price $p$ to use in the evaluation without some
economic analysis relating
the potential revenue increase to the operator from increased capacity
(as measured by the utility) versus the cost to the operator of the
third-party backhaul.
However, such analysis is beyond the scope of this study, although
business models have been considered elsewhere,
e.g.~\cite{SignalsResearch:09}.

What is relevant for this work is to show that our net
utility maximization optimization can incorporate pricing, whatever the correct
pricing is.  As an illustration,
we fix the adoption rate at 5\% and run the micro+femto optimization
 for different prices $p$ in \eqref{eq:Cfemto}.\footnote{It should be noted that loading price $p$ is unitless (as far as the resource assignment algorithm is concerned) and simply represents the weight of the penalty incurred on net utility.}
Fig.~\ref{fig:cdf_vs_rate} plots the resulting rate distributions across
the UEs for the different prices.
As we would expect, as the price is increased,
the number of users scheduled on third-party links along with the volume of offloaded traffic decreases as a result of the algorithm's penalization of such users. The optimal allocation of resources therefore tends to involve these links less and less.
As $p \arr \infty$, the rate distribution approaches the distribution using
only the operator-controlled microcells.

\begin{figure}
\centering
\includegraphics[width=0.8\columnwidth,trim=1.3in 3.4in 1.3in 3.4in]{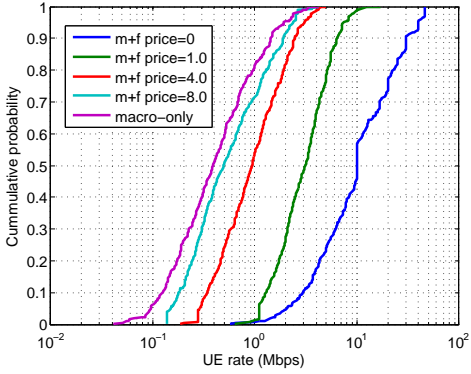}
\caption{Distribution of the UE rates for micro+femto urban real-world
model with femto price
$p \in \{0,1,4,8\}$ and 5\% adoption.  Also, plotted is the micro-only
distribution which corresponds to the price $p=\infty$. }
\label{fig:cdf_vs_rate}
\end{figure}

Also, although we cannot assess the absolute economic value of femtocell offload,
we can conduct the following simple comparison: Suppose the operator
wishes to increase capacity of its network.  We
compare the following two methods:
\begin{itemize}
\item \emph{Increase operator-controlled cells:}
In this method, the operator
does not use any femtocells and increases capacity by adding
operator-controlled micro/macrocells only, i.e.\
traditional \emph{cell splitting}.
To simulate this scenario, we imagine a network with all the parameters
being the same, except that the density of the micro/macrocells is increased
by some factor $\alpha \geq 1$.  Increasing the density in this manner
will incur a variety of costs to the operator including
the capital and operating expenses of the new base stations as well as the
cost of the additional backhaul.
For the moment, we will only consider the additional backhaul costs.
Then, as we vary $\alpha$,
we can estimate the gain in network capacity as a function of the
additional backhaul.

\item \emph{Femto offload:}  As an alternate approach,
we imagine that the operator
adds no new micro/macro base station cells of its own and relies entirely on
purchasing capacity via femto offload.  To simulate this scenario,
we fix an adoption rate at some reasonable value (we assume 5\%), and then
increase network capacity by lowering the relative cost $p$ in \eqref{eq:Cfemto},
from $p=\infty$ (where the operator uses no femtocell offload)
to $p=0$ (where the network
purchases any capacity on femtocells without regard to cost).  Then,
we can again measure the increase in network capacity as a function
of the additional backhaul costs, where the additional backhaul in this case
is on the third-party femtocells.
\end{itemize}

The results of this comparison are shown in Fig. \ref{fig:util_gain_backhaul}.
Plotted is the system utility as a function of the additional backhaul
required for both methods -- adding operator-controlled macro/microcells,
or offloading to femtocells on existing backhaul.  We measure capacity via
the utility. Since we assume a proportional fair utility,
$\sum_i \log(\rbar_i^{MS})$, the utility
is equivalent to the geometric mean rate
$(\prod_i \rbar_i^{MS})^{1/N_{MS}}$.  The geometric mean rate is a better
measure of network capacity than average rate since it penalizes mobiles with
lower rates more significantly. Nevertheless, although it is not plotted,
very similar curves would be observed with either arithmetic mean rate or
cell edge throughput.

\begin{figure}
\begin{subfigure}[t]{1.0\linewidth}
\centering
\includegraphics[width=0.8\columnwidth, trim=1.3in 3.4in 1.3in 3.4in]{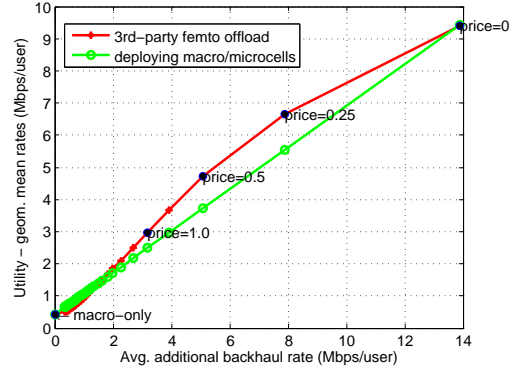}
\caption{Urban scenario}
\label{fig:util_vs_rate_urban}
\end{subfigure}

\begin{subfigure}[t]{1.0\linewidth}
\centering
\includegraphics[width=0.8\columnwidth, trim=1.3in 3.4in 1.3in 3.4in]{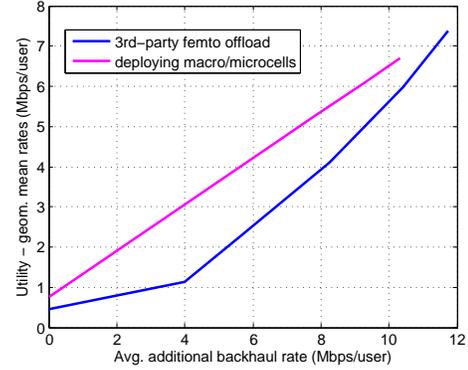}
\caption{Suburban scenario}
\label{fig:util_vs_rate_suburban}
\end{subfigure}

\caption{Average utility as a function of additional backhaul,
comparing adding operator-controlled micro/macrocells only vs.\
relying entirely on femtocell offload. For the femtocell offload case,
we assume a 5\% adoption rate and vary the amount of backhaul
used on the femtocells.
 \label{fig:util_gain_backhaul} }

\end{figure}

We see from Fig.~\ref{fig:util_gain_backhaul} that, for the urban
scenario, increasing the utility requires roughly the same amount of
additional backhaul whether deploying more operator-controlled cells
or using femtocell offload.  In the suburban scenario, the femtocell
offload requires significantly more additional backhaul for small increases
in utility, but requires only modestly more additional backhaul
for larger increases in capacity.

Now, as discussed in the Introduction,
it is likely that the backhaul from femtocell offload would be significantly
lower cost than the purchasing leased lines with guaranteed rates
needed for operator-controlled macro/microcells.
Moreover, adding operator-deployed cells incurs additional costs including site
acquisition, infrastructure expenses and network maintenance
\cite{SignalsResearch:09}.  Nevertheless, quantifying the exact savings
would require further economic analysis.

\section*{Conclusions}

We have presented a model for operators to offset backhaul costs by
leveraging existing capacity from third-parties.
In the proposed model,
third parties install open-access femtocells in their networks and
the cellular operator can then opt to move subscribers onto third party
cells for a fee.
The problem of dynamically assigning users between the third-party
and operator-controlled cells is formulated as an optimization problem.
A dual decomposition algorithm is presented
that is extremely general and can incorporate channel
and interference conditions, traffic demands,
backhaul capacity and access pricing.
To evaluate the model, we considered deployments where the third party
femtocells were co-located with existing WiFi APs.  Due to the large numbers of
WiFi APs relative to base station cell sites, our simulations suggest that
network capacity be significantly increased even if
only a small fraction (say 5\%) of current WiFi owners deploy open-access femtocells.
The gains are particularly large in dense urban areas where our data suggests
there are some 2000 WiFi APs per operator cell.
Our optimization can also incorporate a variety of pricing mechanisms
by the third parties, but determining the correct price will need
analysis beyond the scope of this study.
However, our simulations show that, whatever is the correct price,
the additional backhaul to increase
capacity is similar for both adding more operator-controlled cells or offloading
to third-party femtocells.
Thus, assuming third party backhaul can be offered at a lower rate than
leased lines for operator-controlled cells,
the savings of the proposed method can be significant.
In this way, opportunistic backhaul can offer a scalable, low-cost method to
increase network capacity and address the growing demands on cellular
wireless networks.

\bibliographystyle{IEEEtran}
\bibliography{../bibl}

\end{document}